\documentclass[runningheads]{llncs}
\usepackage[T1]{fontenc}
\usepackage{graphicx}
\usepackage{amssymb, amsmath}
\usepackage{mathtools}

\usepackage{thmtools}
\usepackage{thm-restate} 
\usepackage{cite}
\usepackage{hyperref}
\usepackage[capitalize,noabbrev,nosort]{cleveref}
\usepackage{subcaption}
\usepackage{tikz}
\usetikzlibrary{automata,arrows,calc}
\tikzset{
	state/.append style={semithick,initial text=,inner sep=0,outer sep=0,minimum size=1.5em},
	transition/.append style={->,>=stealth',shorten >=1pt,semithick},
	transition label/.append style={outer sep=3pt},
	every initial by arrow/.append style={initial text=,initial where=,transition},
}

\usepackage{hyperref}
\usepackage{enumitem}
\usepackage{color}

\DeclareMathOperator{\GL}{GL}

\DeclareMathOperator{\diag}{diag}
\DeclareMathOperator{\basin}{basin}
\DeclareMathOperator{\fun}{fun}

\DeclareMathOperator{\lcm}{lcm}
\DeclareMathOperator{\ord}{ord}
\DeclareMathOperator{\Per}{\textnormal{\textsf{Per}}}
\DeclareMathOperator{\Divs}{\textnormal{\textsf{Div}}}
\DeclareMathOperator{\pre}{pre}
\DeclareMathOperator{\rank}{rank}

\newcommand{\C}{\mathbb{C}}

\newcommand{\N}{\mathbb{N}}
\newcommand{\Q}{\mathbb{Q}}
\newcommand{\R}{\mathbb{R}}
\newcommand{\Z}{\mathbb{Z}}

\newcommand{\dfao}{\textsc{DFAO}}
\newcommand{\n}{_{n \ge 0}}
\newcommand{\e}{\boldsymbol{e}}
\newcommand{\s}{\boldsymbol{s}}

\newcommand{\ceil}[1]{\left\lceil #1 \right\rceil}
\newcommand{\size}[1]{\lvert #1 \rvert}

\newcommand{\seq}[1]{\cite[\href{http://oeis.org/#1}{#1}]{OEIS}}

\begin{document}
\title{Magic Numbers in Periodic Sequences}
%
%
\author{
Savinien Kreczman\inst{1}\orcidID{0000-0002-1928-0028} \and \\
Luca Prigioniero\inst{2}\orcidID{0000-0001-7163-4965} \and \\
Eric Rowland\inst{3}\orcidID{0000-0002-0359-8381} \and \\
Manon Stipulanti\inst{1}\thanks{Corresponding author}\orcidID{0000-0002-2805-2465}}
\authorrunning{S.~Kreczman et al.}
%
\institute{
	Department of Mathematics, University of Liège, Liège, Belgium
	\email{$\{$savinien.kreczman,m.stipulanti$\}$@uliege.be}
	\and
	Dipartimento di Informatica, Universit\`a degli Studi di Milano, Milan,  Italy
	\email{prigioniero@di.unimi.it} 
	\and
	Department of Mathematics, Hofstra University, Hempstead,  New York, USA
	\email{eric.rowland@hofstra.edu}
}
\maketitle              

\begin{abstract}
	In formal languages and automata theory,  the magic number problem can be formulated as follows:
	for a given integer $n$,
	is it possible to find a number~$d$ in the range~$[n,2^n]$ such that there is no minimal deterministic finite automaton with~$d$ states that can be simulated by a minimal nondeterministic finite automaton with exactly~$n$ states?
	If such a number $d$ exists, it is called magic.
	In this paper,
	we consider the magic number problem in the framework of deterministic automata with output,
	which are known to characterize automatic sequences.
	More precisely,
	we investigate magic numbers for periodic sequences viewed as either automatic, regular, or constant-recursive.

\end{abstract}
\section{Introduction}\label{section: introduction}

The magic number problem has received wide attention in
formal languages and automata theory
since its recent formulation~\cite{MagicNumber1,MagicNumber2}.
In this setting,
it is well known that,
given a nondeterministic finite automaton with~$n$ states, 
it is always possible to obtain an equivalent deterministic finite automaton,
for example by applying the \emph{powerset construction}.
The number of states of the resulting machine may vary between $n$ and $2^n$,
and the upper bound cannot be reduced in the worst case~\cite{MF71}.
Therefore,
it is natural to ask whether,
given a number~$n$,
it is possible to find a number~$d$ satisfying~$n\leq d \leq 2^n$,
such that there is no minimal deterministic finite automaton with exactly~$d$ states
that can be simulated by an optimal nondeterministic finite automaton with exactly~$n$ states.
If such a number~$d$ exists, it is called~\emph{magic} for~$n$.
This problem has been studied for finite automata with one- and two-letter alphabets
(for which magic numbers have been found~\cite{MagicNumber1,MagicNumber2,Gef07}), 
for larger alphabets (for which it has been proved that
no magic number exists for each $n\ge 0$~\cite{Jir06,Gef05,Jir11}), 
and for some special cases~\cite{HJK12}.
The formulation of the magic number problem has been adapted to the cost of operations on regular languages~\cite{Jir08}.
In that case, the goal is to determine if the whole range of sizes between the lower and the upper bounds can be obtained when applying a given language operation.

\begin{example}
	\label{ex: magic numbers NFA}
	The nondeterministic finite automaton~$\mathcal{A}$ in \cref{fig: ex: magic numbers NFA} has~$5$ states
	and its corresponding minimal deterministic finite automaton,  which can be obtained by minimizing the automaton produced when applying the powerset construction to~$\mathcal{A}$,  has~$2^5 = 32$ states.
	For each~$d \in\{ 5, \ldots, 32\}$,
	it is possible to modify~$\mathcal{A}$,
	by adding suitable transitions on the symbol~$c$,
	to obtain a corresponding minimal deterministic finite automaton with~$d$ states.
	Moreover,
	this example can be generalized to nondeterministic finite automata with~$n$ states,
	for every~$n\ge 0$,
	to prove that no magic numbers exist
	in the case of three-letter alphabets~\cite{Jir11}.

	\begin{figure}[tb]
		\centering
		\def\nodeDistance{5em}
		\begin{tikzpicture}[
			font=\scriptsize,
			node distance=\nodeDistance
			]
			\path[align=center]
			(0,0)
				node[state,initial,accepting, initial where=below]	(5) {}
			++(0:\nodeDistance)
				node[state] 					(4) {}
			++(0:\nodeDistance)	
				node[state]						(3) {}
			++(0:\nodeDistance)
				node[state]						(2) {}
			++(0:\nodeDistance)
				node[state]						(1) {}
				;

			\path[transition,near end, above]
				(1.-195) edge node[] {$a,b$} (2.15)
				(2) edge node[] {$a,b$} (3)
				(3) edge node[] {$a,b$} (4)
				(4) edge node[xshift=.4em] {$a,b,c$} (5)
			;
			\path[transition]
				(2.-15) edge node[near start,yshift=-1ex,xshift=-.5em] {$b$} (1.195)
				(3) edge[bend right=15] node[xshift=-1em,yshift=-.5ex,near start] {$b$} (1.-150)
				(4) edge[bend right=15] node[xshift=-1em,yshift=-.5ex,near start] {$b$} (1.-125)
				(5) edge[bend right=18] node[xshift=-1em,yshift=-.5ex,near start] {$b$} ($(1.-89)+(-1pt,-.5pt)$)
			;
			\path[transition]
				(1) edge[loop above=10] node[right,xshift=.25em,yshift=-1ex] {$b$} (1)
				(5) edge[loop above=10] node[left,xshift=-.25em,yshift=-1ex] {$b,c$} (1)
			;
		\end{tikzpicture}
		\caption{
			A nondeterministic finite automaton with~$5$ states
			having an equivalent minimal deterministic finite automaton with~$2^5$ states.
		}
		\label{fig: ex: magic numbers NFA}
	\end{figure}
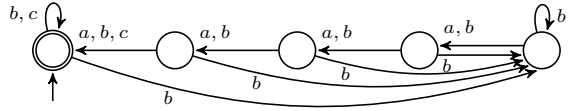
\end{example}

At first,  our intention was to investigate the magic number problem in the case of \emph{deterministic finite automata with output}.
It is known that these machines,  which are simple extensions of deterministic finite automata in which every state is associated with an output,  generate the class of \emph{automatic sequences}.
More precisely,  a sequence $s=s(n)\n$ is automatic if there is a deterministic finite automaton with output that,  for each integer $n$,  outputs $s(n)$ when the digits of $n$ are given as input~\cite[Chapter~5]{Allouche--Shallit-2003}.
First systematically studied by Cobham~\cite{Cobham72},  automatic sequences have become a central topic in combinatorics on words,  with applications in numerous fields,  especially number theory~\cite{Allouche--Shallit-2003}.
There is a well-established correspondence between a specific minimal deterministic finite automaton with output generating the sequence $s$ and the so-called \emph{kernel} of $s$,  which is a set of particular subsequences extracted from $s$.
Thus,  since we are interested in the number of states,  it is equivalent to study the kernel sizes of automatic sequences,  which is called the \emph{rank}.
Specifically,  given a period length $\ell$,  we want to find lower and upper bounds on the possible values of the rank for all periodic sequences with period length $\ell$.
Moreover,  we examine the magic number problem in this context: given an integer $\ell\ge 1$,  is it possible to find numbers~$d$ within the previous bounds such that there is a periodic sequence with period length $\ell$ generated by a deterministic finite automaton with output with exactly~$d$ states or, equivalently,  having rank exactly~$d$?
In some specific case,  we obtain a full characterization of magic numbers for periodic sequences.

We then extended this study to two other families of sequences,
namely \emph{constant-recursive} and \emph{regular}.
A sequence is constant-recursive if each of its elements can be obtained from a linear combination of the previous ones.
The rank of such a sequence is the smallest degree of all linear recurrences producing it.
Motivation to study constant-recursive sequences is that they appear in many areas of computer science and mathematics,
such as theoretical biology,  software verification,  probabilistic model checking, quantum computing,  and crystallography; for further details,  see \cite{Ouaknine-Worrell} or the fine survey~\cite{Rec-seq}, and references therein.
There are still challenging open problems involving these sequences.
For instance,  in general,  it is not known if it is decidable whether or not a constant-recursive sequence takes the value $0$. 
This question is referred to as the \emph{Skolem problem}.
For this class of sequences,  we characterize the corresponding magic numbers as sums of the Euler totient function evaluated at divisors of the period length of the sequence.

On the other hand,
a sequence $s$ is regular if the vector space generated by its kernel is finitely generated~\cite{Allouche--Shallit-1992}. 
In this class,  the rank of a sequence is the dimension of the vector space.
Every automatic sequence,  necessarily on a finite alphabet,  is regular.
In that sense,  the notion of regular sequences generalizes that of automatic sequences to infinite alphabets and thus provides a wider framework for the study of sequences.
Regular sequences naturally arise in many fields,  such as analysis of algorithms, combinatorics, number theory,  numerical analysis, and the theory of fractals~\cite{Allouche--Shallit-1992}.
In some specific case,  we show that the regular case actually boils down to the constant-recursive one.

The paper is organized as follows.
In \cref{sec: prelim}, we introduce the necessary preliminaries.
A specific section is then dedicated to each class of sequences:
namely \cref{section: automatic sequences} to automatic sequences,
\cref{sec: constant-recursive sequences} to constant-recursive sequences,
and \cref{sec: regular sequences} to regular sequences.
We chose to provide the necessary background for each class of sequences in the corresponding section.

\section{Preliminaries}\label{sec: prelim}

Within a family $\mathcal{F}$ of sequences,
we will consider a property $\mathrm{P}_\mathcal{F}$ on numbers.
We say that a number is \emph{magic} if it does not satisfy $\mathrm{P}_\mathcal{F}$;
otherwise,
adopting the terminology from~\cite{Gef07},
it is called \emph{muggle}.
The \emph{range for $\mathrm{P}_\mathcal{F}$} is the smallest interval containing all muggle numbers with respect to $\mathrm{P}_\mathcal{F}$.

The \emph{period length} of a periodic sequence $s$ is the smallest integer $\ell \ge 1$ satisfying $s(n + \ell) = s(n)$ for all $n \ge 0$.
Let $\Per(\ell)$ denote the set of such periodic sequences with period length~$\ell$.
The \emph{period} of a sequence $s \in \Per(\ell)$ is the tuple~$(s(0), s(1), \dots, s(\ell - 1))$.
Note that,  in the following,  we discard the case~$\ell=1$,  which corresponds to constant sequences and for which the three families give a corresponding rank equal to~$1$.

For a periodic sequence~$s$,
subsequences of the form~$s(c n + r)\n$ are also periodic,
and their period lengths satisfy the following.

\begin{restatable}{proposition}{SubseqPeriodic} 
	\label{pro: subsequence period length}
	Let $\ell\ge 2$, $c,r\ge 0$ be integers.
	If $s$ is a periodic sequence whose period length divides $\ell$,
	then $s(c n + r)\n$ is periodic with period length dividing~$\frac{\ell}{\gcd(c, \ell)}$.
\end{restatable}
\begin{proof}
\label{proof: pro: subsequence period length}
For all $n \ge 0$, we have
\[
	s\!\left(c \left(n + \frac{\ell}{\gcd(c, \ell)}\right) + r\right)
	= s\!\left(c n + \frac{c}{\gcd(c, \ell)} \ell + r\right)
	= s(c n + r)
\]
since $\frac{c}{\gcd(c, \ell)}$ is an integer and the period length of $s$ divides $\ell$,  which suffices.
\qed
\end{proof}

\section{Automatic Sequences}\label{section: automatic sequences}
In this section, we are interested in representations of periodic sequences as automatic sequences.
We briefly mention earlier works concerning \emph{periodicity} and \emph{automaticity}.
\emph{Eventual periodicity} and \emph{automatic sets} have been considered in~\cite{Honkala-1986, Boigelot--Mainz--Marsault--Rigo-2017, Marsault-2019}.
For $\ell\ge 1$,  Alexeev~\cite{Alexeev} and Sutner~\cite{Sutner} studied the automata generating the sequence with period~$(1,0,0,\ldots,0)$,
with $\ell-1$ trailing zeroes.
Sutner and Tetruashvili determined the rank of the sequence with period $(0, 1, \dots, \ell - 1)$ for specific values of~$\ell \ge 1$~\cite[Claim~2.1]{Sutner--Tetruashvili}.
Bosma studied some particular cases of several families of periodic sequences on a binary alphabet and considered upper and lower bounds on the rank~\cite{Bosma},
while Zantema gave an upper bound on the rank of all periodic sequences~\cite[Theorem~10]{Zantema}.

We now introduce the definitions for this section.
Let $k,\ell\ge 2$ be integers.
Define $\ord_\ell(k)$ to be the smallest integer $m\ge 1$ such that $k^{e+m} \equiv k^e \mod{\ell}$ for all sufficiently large $e$.
In the special case when~$k$ and~$\ell$ are coprime, $\ord_\ell(k)$ is the usual multiplicative order of $k$ modulo $\ell$.
Let $\pre_\ell(k)$ denote the smallest integer $e\ge 0$ such that $k^{e+\ord_\ell(k)} \equiv k^e \mod{\ell}$.
In other words, $\pre_\ell(k)$ is the length of the preperiod of the sequence $(k^e \bmod \ell)_{e \ge 0}$,
and $\ord_\ell(k)$ is the (eventual) period length.
For every integer $n$, we let $\Divs(n)$ be the set of divisors of $n$.

Given an integer $k$,
a sequence $s$ is \emph{$k$-automatic} if there is a deterministic finite automaton with output (\dfao) that,
for each integer~$n$,  outputs $s(n)$ when fed with the base-$k$ digits of $n$.
To tie the automaton to the base, we call it a \emph{$k$-automaton}. 
Here,
automata read \mbox{base-$k$} representations beginning with the least significant digit
(note that the other reading direction produces the same set of sequences).
There is also a characterization of automatic sequences in terms of the so-called \emph{kernel}. 
In particular,
the \emph{$k$-kernel} of a sequence $s$ is the set of subsequences $\ker_k(s)=\{s(k^e n+ j )\n \colon \text{$e \ge 0$ and $0\le j \le k^e-1$} \}$.
A sequence is $k$-automatic if and only if its $k$-kernel is finite.
We define the \emph{$k$-rank}
(or simply the \emph{rank} when the context is clear)
of a~{$k$-automatic} sequence $s$ to be the quantity $\rank_k(s) = \size{\ker_k(s)}$.
In fact, $\ker_k(s)$ is in bijection with the set of states in the minimal $k$-automaton that generates $s$ with the property that leading $0$'s do not affect the output.
Every periodic sequence is $k$-automatic for every $k \ge 2$~\cite[p.~88]{Buchi}.
Note that,
in this case,
the alphabet does not matter,
so,
without loss of generality,
the sequences in this section will have nonnegative integer values.
For further details on automatic sequences see, e.g.,~\cite{Allouche--Shallit-2003}.

\begin{example}\label{ex: 2-6 kernel}
	For the sequence $s\in\Per(6)$ with period $(0,1,2,3,4,5)$,
	we illustrate the relationship between the minimal $2$-automaton for $s$ and its $2$-kernel in \cref{fig: DFAO intro example} and \cref{tab:example-kernel-aut}.
\begin{table}[h]
	\caption{
		For $s\in\Per(6)$ with period $(0,1,2,3,4,5)$,
		its $2$-kernel contains  $13$ distinct sequences shown in the second column.
		Two pairs $(e,j)$ and $(e',j')$ are equivalent if and only if  $s(k^e n+ j )\n=s(k^{e'} n+ j')\n$.
		Therefore,  in the third column,  the state in the $2$-automaton for $s$ of \cref{fig: DFAO intro example} is labeled $ej$ where $(e,j)$ is the lexicographically smallest representative for this equivalence relation.%
	}
	\label{tab:example-kernel-aut}
	\[
		\begin{array}{l|l|c}
			(e,j) & s((2^e \bmod 6) n + (j \bmod 6))\n & \text{state} \\
			\hline
			(0,0)						& s(n)\n = 0, 1, 2, 3, 4, 5, 0, 1, 2, 3, 4, 5,\ldots& 00 \\
			(1,0), (3,0), (3,6)			& s(2n)\n = 0, 2, 4, 0, 2, 4, \ldots				& 10	\\
			(1,1), (3,1), (3,7)			& s(2n+1)\n = 1, 3, 5, 1, 3, 5, \ldots				& 11	\\
			(2,0), (4,0), (4,6), (4,12) & s(4n)\n = 0, 4, 2, 0, 4, 2, \ldots				& 20	\\
			(2,1), (4,1), (4,7), (4,13) & s(4n+1)\n = 1, 5, 3, 1, 5, 3, \ldots				& 21	\\
			(2,2), (4,2), (4,8), (4,14) & s(4n+2)\n = 2, 0, 4, 2, 0, 4, \ldots				& 22	\\
			(2,3), (4,3), (4,9), (4,15) & s(4n+3)\n = 3, 1, 5, 3, 1, 5, \ldots				& 23	\\
			(3,2) 						& s(2n+2)\n = 2, 4, 0, 2, 4, 0, \ldots				& 32	\\
			(3,3) 						& s(2n+3)\n = 3, 5, 1, 3, 5, 1, \ldots				& 33	\\
			(3,4) 						& s(2n+4)\n =4, 0, 2, 4, 0, 2, \ldots				& 34	\\
			(3,5) 						& s(2n+5)\n = 5, 1, 3, 5, 1, 3, \ldots				& 35\\
			(4,4), (4,10) 				& s(4n+4)\n = 4, 2, 0, 4, 2, 0, \ldots				& 44\\ 
			(4,5), (4,11) 				& s(4n+5)\n = 5, 3, 1, 5, 3, 1, \ldots				& 45\\
		\end{array}
	\]
\end{table}
	\begin{figure}[h]
		\centering
		\def\nodeDistance{5em}
		\begin{tikzpicture}[
			font=\scriptsize,
			inner sep=.7ex,
			node distance=\nodeDistance
			]
			\path[text width={width("$45 / 5$")},align=center]
				(0,0) 					node[state,initial,initial where=below](0) {$00/ 0$}
				++(-42:\nodeDistance)node[state]	(1) {$10 / 0$}
				++(30:\nodeDistance)	node[state]	(3) {$20 / 0$}
				++(0:\nodeDistance)		node[state] (9)	{$34 / 4$}
				++(-30:\nodeDistance)	node[state] (11){$44 / 4$}

				(1)
				++(-30:\nodeDistance)	node[state]	(5) {$22 / 2$}
				++(0:\nodeDistance)		node[state] (7)	{$32 / 2$}

				(0)
				++(180+42:\nodeDistance)	node[state] (2) {$11 / 1$}	
				++(180-30:\nodeDistance)	node[state]	(4) {$21 / 1$}
				++(180:\nodeDistance)		node[state] (10){$35 / 5$}
				++(180+30:\nodeDistance)	node[state] (12){$45 / 5$}

				(2)
				++(180+30:\nodeDistance)	node[state]	(6) {$23 / 3$}
				++(180:\nodeDistance)		node[state] (8)	{$33 / 3$}
				;

			\path[transition]
				(0) edge node[above] {$0$} (1)
				(0) edge node[above] {$1$} (2);

			\path[transition,every node/.append style={near end}]
				(1) edge[bend left=10] node[above] {$0$} (3)
				(1) edge[bend left=10] node[above] {$1$} (5)

				(3) edge[bend left=10] node[below] {$0$} (1)
				(3) edge[bend left=10] node[above] {$1$} (9)
				
				(5) edge[bend left=10] node[above] {$0$} (7)
				(5) edge[bend left=10] node[below] {$1$} (1)
				
				(7) edge[bend left=10] node[below] {$0$} (5)
				(7) edge[bend left=10] node[above] {$1$} (11)
				
				(9) edge[bend left=10] node[above] {$0$} (11)
				(9) edge[bend left=10] node[below] {$1$} (3)
				
				(11) edge[bend left=10] node[below] {$0$} (9)
				(11) edge[bend left=10] node[below] {$1$} (7)

				(2) edge[bend left=10] node[below] {$0$} (4)
				(2) edge[bend left=10] node[below] {$1$} (6)

				(4) edge[bend left=10] node[above] {$0$} (2)
				(4) edge[bend left=10] node[below] {$1$} (10)
				
				(6) edge[bend left=10] node[below] {$0$} (8)
				(6) edge[bend left=10] node[above] {$1$} (2)
				
				(8) edge[bend left=10] node[above] {$0$} (6)
				(8) edge[bend left=10] node[below] {$1$} (12)
				
				(10) edge[bend left=10] node[below] {$0$} (12)
				(10) edge[bend left=10] node[above] {$1$} (4)
				
				(12) edge[bend left=10] node[above] {$0$} (10)
				(12) edge[bend left=10] node[above] {$1$} (8)
			;
		\end{tikzpicture}
		\caption{The minimal $2$-automaton outputting the sequence with period $(0,1,2,3,4,5)$.  Each state has a label of the form $\alpha/\rho$ where $\alpha$ corresponds to a state name in the third column of \cref{tab:example-kernel-aut} and $\rho$ is the output.}
		\label{fig: DFAO intro example}
	\end{figure}
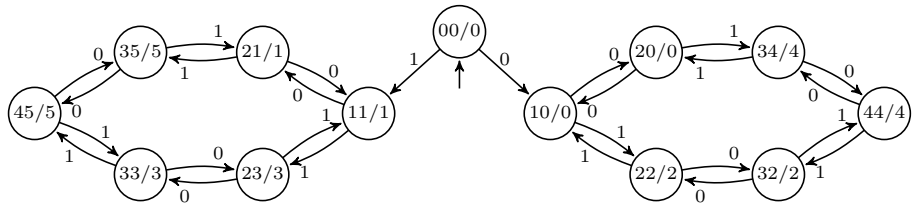
\end{example}

\subsection{Bounds on Magic Numbers}

For integers $k,\ell\ge 2$ and the family of periodic sequences $\Per(\ell)$,
we consider the following property:
an integer satisfies $\mathrm{P}_{\text{a}}(k,\ell)$ if and only if it is equal to the $k$-rank of some sequence in $\Per(\ell)$.
By merging \cref{pro: upper bound on rank for k-automatic sequences} and \cref{lem: coprime kernel},  
we obtain the range for $\mathrm{P}_{\text{a}}(k,\ell)$ when $k$ and $\ell$ are coprime as stated in below.

\begin{corollary}
\label{cor: upper bound tight aut}
Let $k,\ell\ge 2$ be coprime integers.
The interval $[\ell,B_\ell(k)]$ is the range for $\mathrm{P}_{\text{a}}(k,\ell)$ where $B_\ell(k):= \left(\sum_{e=0}^{\pre_\ell(k) - 1} \min(k^e,\ell)\right) + \ell \cdot \ord_\ell(k)$.
\end{corollary}

We first obtain the upper bound in the general case.

\begin{restatable}{proposition}{UpperBoundAutomatic}
	\label{pro: upper bound on rank for k-automatic sequences}
	Let~$k,\ell\ge 2$ be two integers.
	The $k$-rank of every sequence in~$\Per(\ell)$ is at most $B_\ell(k)$.
	Moreover,
	the $k$-rank of the periodic sequence with period $(0, 1, \dots, \ell - 1)$ is $B_\ell(k)$. 
\end{restatable}%

To prove \cref{pro: upper bound on rank for k-automatic sequences},  we need a new notion and two lemmas.
We start off with a first one.

\begin{lemma}
\label{lem: rewriting}
Let $k,\ell \ge 2$ be integers and let $s \in \Per(\ell)$.
Every sequence in $\ker_k(s)$ can be written as $s(cn+r)\n$ with $c,r\in\{0,1,\ldots,\ell-1\}$.
In particular,  if $k$ and~$\ell$ are coprime, then $c$ is invertible in $\Z_\ell$.
\end{lemma}
\begin{proof}
Let $e \ge 0$ and $0\le j \le k^e-1$.
Since $s\in\Per(\ell)$,  we have $s(k^e n + j)\n = s((k^e \bmod{\ell} ) n + (j \bmod{\ell}))\n$.
Now suppose that $k$ and $\ell$ are coprime.
We show that $\gcd(c,\ell)=1$.
To the contrary, suppose that $\ell$ and $c$ have a common prime factor $p$, and write $\ell=m_1 p$ and $c=m_2 p$ for some integers $m_1,m_2\ge 1$.
By definition, let $e,m_3\ge 0$ be such that $k^e=m_3\ell + c$.
Then we have $k^e=(m_1m_3+m_2)p$, and thus $p$ is a common divisor of $k$ and $\ell$, which is impossible.
\qed
\end{proof}

For a sequence $s \in \Per(\ell)$, we define two sets of sequences: the \emph{$(k,\ell)$-funnel} and the \emph{$(k,\ell)$-basin}:
\begin{align*}
\fun_{k,\ell}(s)&=\{s(k^e n + j)\n : 0\le e\le \pre_\ell(k) - 1 \text{ and } 0\le j \le \min(k^e,\ell)-1\}, \\
\basin_{k,\ell}(s)&=\{s(k^{e + \pre_\ell(k)} n + j)\n : 0\le e\le \ord_\ell(k) - 1 \text{ and } 0\le j \le \ell-1\}.
\end{align*}

It is clear from the definition that $\fun_{k,\ell}(s) \subseteq \ker_k(s)$.
We will show that $\basin_{k,\ell}(s) \subseteq \ker_k(s)$ as well.
For each kernel sequence $s(k^{e + \pre_\ell(k)} n + j)\n$ in $\basin_{k,\ell}(s)$, the kernel sequence $s(k^{e + \pre_\ell(k)} (k n + r) + j)\n$ is also in $\basin_{k,\ell}(s)$, so the basin is closed under this operation.
The funnel consists of kernel sequences $s(k^e n + j)\n$ for which $e$ is too small for $k^e$ to belong to the eventual period of powers of $k$ modulo $\ell$.
It follows from the definition of the basin that if $s\in\Per(\ell)$ and $k\equiv k' \mod\ell$, then $\basin_{k,\ell}(s)=\basin_{k',\ell}(s)$.

We have the following description of the kernel.

\begin{lemma}\label{lem: kernel decomposition}
Let $k,\ell\ge 2$ be integers.
If $s \in \Per(\ell)$, then $\ker_k(s) = \fun_{k,\ell}(s) \cup \basin_{k,\ell}(s)$. 
\end{lemma}

\begin{proof}
We write $\ker_k(s)$ as the union of the two sets
\begin{align*}
	F_{k,\ell}(s)&=\{s(k^e n + j)\n : \text{$0\le e\le \pre_\ell(k) - 1$ and $0\le j \le k^e-1$}\}, \\
	B_{k,\ell}(s)&=\{s(k^e n + j)\n : \text{$e \ge \pre_\ell(k)$ and $0\le j \le k^e-1$}\}.
\end{align*}
Since $s \in \Per(\ell)$ (or \cref{lem: rewriting}), we have
\begin{align*}
	F_{k,\ell}(s) &=\{s(k^e n + j)\n : \text{$0\le e\le \pre_\ell(k) - 1$ and $0\le j \le \min(k^e,\ell)-1$}\}, \\
	B_{k,\ell}(s) &=\{s(k^e n + j)\n : \text{$e \ge \pre_\ell(k)$ and $0\le j \le \min(k^e,\ell)-1$}\}.
\end{align*}
Therefore, $F_{k,\ell}(s) = \fun_{k,\ell}(s)$.
We show $B_{k,\ell}(s) = \basin_{k,\ell}(s)$.
There are two cases.
If $k^{\pre_\ell(k)} \ge \ell$, then $\min(k^e,\ell) = \ell$ for all $e \ge \pre_\ell(k)$, so
\begin{align*}
	B_{k,\ell}(s)
	&= \{s(k^e n + j)\n : \text{$e \ge \pre_\ell(k)$ and $0\le j \le \ell-1$}\} \\
	&= \{s(k^{e + \pre_\ell(k)} n + j)\n : \text{$e \ge 0$ and $0\le j \le \ell-1$}\}.
\end{align*}
Since $k^{e + \ord_\ell(k)} \equiv k^e \mod \ell$ for all $e \ge \pre_\ell(k)$, we have $B_{k,\ell}(s) = \basin_{k,\ell}(s)$.
Now if $k^{\pre_\ell(k)} < \ell$, let $E = \ceil{\log_k \ell}$, so that $k^e \ge \ell$ for all $e \ge E$.
Since $\pre_\ell(k) < E$ (by definition of $E$),
\begin{multline*}
	B_{k,\ell}(s)
	= \{s(k^e n + j)\n : \text{$\pre_\ell(k) \leq e \leq E - 1$ and $0\le j \le k^e-1$}\} \\
	\cup \{s(k^e n + j)\n : \text{$e \ge E$ and $0\le j \le \ell-1$}\}.
\end{multline*}
We show that the first set in the right-hand side of the previous equality is a subset of the second.
Let $\pre_\ell(k) \leq e \leq E - 1$ and $0\le j \le k^e-1$, so that $s(k^e n + j)\n$ belongs to the first set.
Let $e' \ge E$ such that $k^{e'} \equiv k^e \mod \ell$.
Then $s(k^e n + j)\n = s(k^{e'} n + j)\n$.
Therefore,
\begin{align*}
	B_{k,\ell}(s)
	&= \{s(k^e n + j)\n : \text{$e \ge E$ and $0\le j \le \ell-1$}\} \\
	&= \{s(k^e n + j)\n : \text{$\pre_\ell(k) \le e\le \pre_\ell(k) + \ord_\ell(k) - 1$ and $0\le j \le \ell-1$}\} \\
	&= \basin_{k,\ell}(s).
\end{align*}
This concludes the proof.
\qed
\end{proof}

\begin{proof}[of  \cref{pro: upper bound on rank for k-automatic sequences}]
By \cref{lem: kernel decomposition},
\[
	\rank_k(s) \le
	\size{\fun_{k,\ell}(s)} + \size{\basin_{k,\ell}(s)}
	\leq \left(\sum_{e=0}^{\pre_\ell(k) - 1} \min(k^e,\ell)\right) + \ell \cdot \ord_\ell(k)\text,
\]
as desired.

As for the second part of the statement,  let $g_\ell$ be the sequence with period $(0,1,  \ldots, \ell-1)$. 
The first part of the statement implies that $\rank_k(g_\ell) \leq B_\ell(k)$.
We show that equality holds.
Let $s_1,s_2\in \ker_k(g_\ell)$, and write 
\[
s_1(n)\n=g_\ell(k^{e_1}n+j_1)\n=g_\ell(c_1n+r_1)\n
\]
and 
\[
s_2(n)\n=g_\ell(k^{e_2}n+j_2)\n=g_\ell(c_2n+r_2)\n
\]
with $0\le e_1,e_2\le \pre_\ell(k)+ \ord_\ell(k) -1$,  $0\le j_1,j_2 \le \ell-1$,  and $c_1,c_2,r_1,r_2\in\{0,1,\ldots,\ell-1\}$ by \cref{lem: rewriting}.
By definition of $g_\ell$, we have $g_\ell(n)=g_\ell(m)$ if and only if $n\equiv m \mod{\ell}$.
Therefore,  $s_1=s_2$ implies in particular that $r_2 \equiv r_1 \mod\ell$ ($n=0$) and $c_2+r_2 \equiv c_1+r_1 \mod\ell$ ($n=1$).
Since $c_1,c_2,r_1,r_2\in\{0,1,\ldots,\ell-1\}$, we obtain $r_1=r_2$ and $c_1=c_2$.
This in turn implies that $e_1=e_2$ and $j_1=j_2$,  which concludes the proof.
\qed
\end{proof}

\begin{example}
Let $k=2$ and $\ell=6$.
The sequence $(k^e \bmod{\ell})_{e\ge 0}$ is $1,2,4,2,4,2,4,\ldots$, so $\pre_\ell(k)=1$ and $\ord_\ell(k)=2$.
For the sequence in~\cref{ex: 2-6 kernel},  we obtain $B_\ell(k)=2^0 + 6\cdot 2=13$. 
The minimal DFAO generating~$g_\ell(n)_{n\ge0}$ is shown in \cref{fig: DFAO intro example} and has 13 states.
\end{example}

We then obtain the lower bound of \cref{cor: upper bound tight aut}.

\begin{restatable}{proposition}{LowerBoundAutomatic}
\label{lem: coprime kernel}
	Let $k,\ell\ge 2$ be coprime integers.
	The $k$-rank of every sequence in~$\Per(\ell)$ is at least~$\ell$.
	Moreover,
	the $k$-rank of the periodic sequence with period $(0,0,\ldots,0,1)$ is $\ell$. 
\end{restatable}
\begin{proof}
By \cref{lem: kernel decomposition}, $\ker_k(s) = \basin_{k,\ell}(s)$.
Since $\pre_\ell(k)=0$,  then the sequences $s(n+j)\n$ are elements of $\ker_k(s)$ for $j\ge 0$.
There are $\ell$ distinct such sequences, so $\rank_k(s)\ge\ell$.

For the second part of the statement,  let~$s \in \Per(\ell)$ be the sequence with period $(0,0,\ldots,0,1)$.
By \cref{lem: rewriting,lem: kernel decomposition}, each sequence in $\ker_k(s)  = \basin_{k,\ell}(s)$ can be written $s(cn+r)\n$ where $c,r\in\{0,1,\ldots,\ell-1\}$ and $c\neq 0$ is invertible in~$\Z_\ell$.
The sequence $s(cn+r)\n$ is periodic with period length at most $\ell$.
For all $n\ge 0$, $s(cn+r)=1$ if and only if $cn+r \equiv \ell-1 \mod\ell$, which is equivalent to $n \equiv c^{-1}(\ell-1-r) \mod\ell$.
Therefore,  $s(cn+r)\n$ is a shift of $s$, and $\rank_k(s)\le\ell$.
This inequality together with the first part of the statement gives $\rank_k(s)=\ell$.
\qed
\end{proof}

\subsection{Towards a Characterization}

The magic and muggle numbers with respect to $\mathrm{P}_{\text{a}}(k,\ell)$ when $k$ and $\ell$ are coprime are characterized in the following theorem.

\begin{theorem}\label{thm: characterization of sizes for automatic sequences coprime case}
	Let $k,\ell\ge 2$ be coprime integers.
	Let $A= \{d \ell : d \in \Divs(\ord_\ell(k))\}$.
	Then $A$ is the set of muggle numbers in the range for $\mathrm{P}_{\text{a}}(k,\ell)$.
	In particular,
	the set of magic numbers with respect to $\mathrm{P}_{\text{a}}(k,\ell)$ is $\N\setminus A$.
\end{theorem}
The following lemmas are
useful for the proof of \cref{thm: characterization of sizes for automatic sequences coprime case}.

\begin{restatable}{lemma}{KernelRewritingCoprime}
	\label{lem: kernel rewriting}
	Let $k,\ell\ge 2$ be coprime integers and let $s\in \Per(\ell)$.
	Then $\ker_k(s)=\{s(k^e n+ j )\n \colon \text{$e \ge 0$ and $0\le j \le \ell-1$} \}$.
\end{restatable}
\begin{proof}
	\label{proof: lem: kernel rewriting}
It follows from \cref{lem: kernel decomposition} as,  when $k$ and $\ell$ are coprime,  the set $\fun_{k,\ell}(s)$ is empty.
\qed
\end{proof}

\begin{restatable}{lemma}{SubsequenceCoprime}
	\label{lem: subsequence coprime}
	Let $k,\ell\ge 2$ be coprime integers. 
	If $s$ belongs to $\Per(\ell)$,  then so does~$s(kn)\n$.
\end{restatable}\begin{proof}
	\label{proof: lem: subsequence coprime}
	We proceed by contradiction,  and we assume that $s(kn)\n\in \Per(d)$ with $d \mid \ell$ and $d\neq \ell$.
	Let $j$ be the inverse of $k$ modulo~$\ell$ which exists because $\gcd(k,\ell)=1$.
	Then $s(j k n)\n$ has period length which divides $d$.
	This is impossible since the previous sequence is in fact $s$,  which has period $\ell > d$.
\qed
\end{proof}

We are now ready to give a proof of \cref{thm: characterization of sizes for automatic sequences coprime case}.

\begin{proof}[of \cref{thm: characterization of sizes for automatic sequences coprime case}]
	For the sake of conciseness set $T=\{\rank_k(s) : s \in \Per(\ell)\}$.
	We prove $T=A$ by showing two inclusions.

	\textbf{First inclusion.}
		By \cref{lem: coprime kernel}, $\rank_k(s) \ge \ell$ for all $s \in \Per(\ell)$.
		Using \cref{lem: kernel rewriting},  for each $e \ge 0$, let $B_{k,\ell,e}(s) = \{s(k^e n + j)\n : 0\le j \le \ell-1\}$ be the set of kernel sequences arising from exponent $e$.
		Since $k$ and $\ell$ are coprime,
		the sequences in $B_{k,\ell,e}(s)$ are precisely the $\ell$ shifts $s(k^e (n + i))\n$ ($0\le i \le \ell-1$) of $s(k^e n)\n$, and they are distinct by \cref{lem: subsequence coprime},  so $\size{B_{k,\ell,e}(s)} = \ell$.
		Therefore,  if $B_{k,\ell,e}(s)\cap B_{k,\ell,e'}(s)\neq\emptyset$ for some $e,e' \ge 0$, then $B_{k,\ell,e}(s)=B_{k,\ell,e'}(s)$.
		The sequence $(B_{k,\ell,e}(s))_{e \ge 0}$ of sets is periodic with period length dividing $\ord_\ell(k)$.
		Therefore,  $T \subseteq A$.
		%
		%

	\textbf{Second inclusion.}
		Let $d \in \Divs(\ord_\ell(k))$.
		We consider the mapping $n \mapsto k^d n$ on $\Z/(\ell\Z)$, which is a permutation because $k^d$ is invertible modulo $\ell$.
		This permutation has finite order, and the orbit of $n$ under this mapping is the set $\{n, k^d n, k^{2d} n, \ldots\}$.
		In particular, the orbit of $0$ is $\{0\}$.
		Let $X_0=\{0\}, X_1, \ldots,X_m$ be the distinct orbits.
		These sets  partition $\Z/(\ell\Z)$ (for an illustration,  see \cref{ex: partition}).
		Define the sequence $s$ on the alphabet $\{0,1,\ldots,m\}$ by
		\begin{equation}\label{eq: special-s-dl}
		s(n)=i \quad \text{ where } (n\bmod{\ell})\in X_i.
		\end{equation}
		Then $s$ is periodic since $s(n)=s(n+\ell)$.
		Moreover, $s$ is in $\Per(\ell)$ since $0$ only occurs once in the period, namely at positions congruent to $0$ modulo $\ell$.
		To prove that $A\subseteq T$, it remains to show that $\rank_k(s)=d\ell$.

		First, the sequence $(B_{k,\ell,e}(s))_{e \ge 0}$ defined previously is periodic with period length at most $d$ since $B_{k,\ell,e}(s)=B_{k,\ell,e+d}(s)$ by definition of the orbits.
		We prove that the period length is exactly $d$.
		Let $e,e'\in\{0,1,\ldots,d-1\}$ such that $B_{k,\ell,e}(s) = B_{k,\ell,e'}(s)$; we show $e = e'$.
		Since $0$ only occurs once in the period of $s$, there is precisely one sequence in $B_{k,\ell,e}(s)$ whose initial term is $0$, namely $s(k^e n)\n$.
		Therefore, $s(k^e n)\n = s(k^{e'} n)\n$.
		In particular, $s(k^{e})=s(k^{e'})$,  so $k^{e},k^{e'}\in X_i$ for some $i \neq 0$.
		By definition of $X_i$, $k^{e+jd} \equiv k^{e'} \mod\ell$ for some $j\ge 0$.
		So $e+jd \equiv e' \mod \ord_\ell(k)$, which in turn implies $e \equiv e' \mod d$,  yielding in turn $e = e'$.
		Therefore,  $\rank_k(s)=d\ell$,  as desired.
		
		The second part of the statement follows by \cref{cor: upper bound tight aut}.
		\qed
\end{proof}

Notice that the previous proof is constructive,
that is,
for each divisor $d$ of $\ord_\ell(k)$,
Equation~\eqref{eq: special-s-dl} defines a sequence whose rank is~$d\ell$.

\begin{table}[t]
	\caption{For each possible value of the rank of periodic sequences in $\Per(7)$,  we provide a sequence with such rank.}
	\label{tab:values-rank-aut}
	\[
		\begin{array}{ccccc}
			d & \text{rank} & 3^d \bmod{7} & \text{partition of } \Z/(7\Z) & \text{period}\\
			\hline
			6 & 42 & 1 & \{0\},\{1\},\{2\},\{3\},\{4\},\{5\},\{6\} & (0,1,2,3,4,5,6) \\
			3 & 21 & 6 & \{0\},\{1,6\},\{2,5\},\{3,4\} & (0,1,2,3,3,2,1) \\
			2 & 14 & 2 & \{0\},\{1,2,4\},\{3,5,6\} & (0,1,1,2,1,2,2) \\
			1 & 7 & 3 & \{0\},\{1,2,3,4,5,6\} & (0,1,1,1,1,1,1)
		\end{array}
	\]
\end{table}

\begin{example}\label{ex: partition}
	We have $\ord_7(3)=6$,  so the muggle numbers are in the set $\{7d : d \in \Divs(\ord_7(3))\}=\{7, 14, 21, 42\}$.
	We provide the period of a periodic sequence in~$\Per(7)$ with each rank as shown in \cref{tab:values-rank-aut}.
\end{example}

\begin{remark}
	Neder~\seq{A217519} conjectured that if $s_\ell$ is the sequence with period $(0,1,\ldots,\ell-1)$,  then
	``$\rank_2(s_{\ell}) \leq \ell (\ell-1)$, with equality if and only if $\ell$ is a prime with primitive root $2$''. 
	Recall that a \emph{primitive root} of a prime $\ell$ is an integer $k$ such that $\ord_\ell(k)=\ell-1$.
	Neder's conjecture follows from \cref{thm: characterization of sizes for automatic sequences coprime case} and holds even when $k\neq 2$.
\end{remark}

\section{Constant-Recursive Sequences}\label{sec: constant-recursive sequences}

An integer sequence $s$ is \emph{constant-recursive} if there exist~$c_0, c_1, \dots, c_{d - 1} \in \Z$,
with~$d\in\N$ and~$c_0\neq 0$ such that
\begin{equation}\label{eq: constant-recursive recurrence}
	s(n + d) = c_{d - 1} s(n + d - 1) + \dots + c_1 s(n + 1) + c_0 s(n)
\end{equation}
for all $n \ge 0$.
For a constant-recursive sequence $s$,  we let $\rank(s)$ denote its \emph{rank},
which is the smallest integer $d$ such that $s$ satisfies a recurrence relation of the form of Equation~\eqref{eq: constant-recursive recurrence}.
In other words,  if the set $\{s(n+ j )\n \colon j\ge 0\}$ containing all shifts of 
$s$ is called the \emph{$1$-kernel} of $s$,  then the rank of $s$ is the dimension of the $\Q$-vector space generated by its $1$-kernel.

For an integer $\ell\ge 2$ and the family of periodic sequences $\Per(\ell)$,
we consider the following property:
an integer satisfies $\mathrm{P}_{\text{cr}}(\ell)$ if and only if it is equal to the rank of some sequence in~$\Per(\ell)$.
The first few possible values of $\rank(s)$ for $s\in\Per(\ell)$, with~$\ell \in\{ 2,\ldots,15\}$, are given in \cref{tab:values-rank-cr}.
\begin{table}
	\caption{First few values of $\rank(s)$ when $s\in\Per(\ell)$ for~$\ell \in\{ 2,\ldots,15\}$.}
	\label{tab:values-rank-cr}
	\[
		\begin{array}{c|c||c|c}
			\ell & \rank(s),  s \in \Per(\ell) & \ell & \rank(s),  s \in \Per(\ell) \\ \hline
			2 & 1,2 & 9 & 6,7,8,9 \\
			3 & 2,3 &  10 & 4,5,6,8,9,10 \\
			4 & 2,3,4 & 11 & 10,11 \\
			5 & 4,5 &  12 & 4,5,6,7,8,9,10,11,12 \\
			6 & 2,3,4,5,6 &  13 & 12,13 \\
			7 & 6,7 & 14 & 6,7,8,12,13,14 \\
			8 & 4,5,6,7,8 & 15 & 6,7,8,9,10,11,12,13,14,15
		\end{array}
	\]
\end{table}

The strategy adopted for~$\mathrm{P}_{\text{cr}}(\ell)$ is the following: we provide first the range in \cref{pro: upper bound tight const-rec},  then a full characterization of the magic and muggle numbers  in \cref{thm: allowed dimensions}.
To state them, 
we need some definitions and results on~\emph{cyclotomic polynomials}.
For further details on this topic we address the reader to the nice reference~\cite{Garrett}.
\begin{definition}
	Let $m\ge 1$.
	The \emph{$m$th cyclotomic polynomial} is the monic polynomial $\Phi_m$ whose roots in $\mathbb{C}$ are the primitive $m$th roots of unity.
\end{definition}

The first few cyclotomic polynomials are $\Phi_1(x)=x-1$,  $\Phi_2(x)=x+1$,  $\Phi_3(x)=x^2+x+1$,  $\Phi_4(x)=x^2+1$,  $\Phi_5(x)=x^4+x^3+x^2+x+1$,  and $\Phi_6(x)=x^2-x+1$.
For all $m\ge 2$, $\Phi_m(x)$ is palindromic.
We also have the following identity: $x^n-1 =\prod_{m \in \Divs(n)} \Phi_m(x)$ for all $n\ge 1$.
The \emph{Euler totient function}~$\phi$ maps every positive integer~$n$ to the number of positive integers less than~$n$ and relatively prime to $n$.

\begin{lemma}[Garrett~\cite{Garrett}]\label{lem: degree cyclotomic}
	For all $m\ge 1$, $\Phi_m(x)$ is monic of degree~$\phi(m)$.
\end{lemma}

The \emph{additive version} $\psi$ of the Euler totient function $\phi$ is defined as follows:
\begin{itemize}[nosep]
	\item for every prime $p$ and every $k\ge 1$, $\psi(p^k) = \phi(p^k)$;
	\item if $n$ is odd,  then $\psi(2n)=\phi(n)$; if $m$ and $n$ are relatively prime and not equal to $2$, then $\psi(m n)=\phi(m)+\phi(n)$.
\end{itemize}

\begin{restatable}{proposition}{UpperBoundCR}
	\label{pro: upper bound tight const-rec}
	Let $\ell\ge 2$.
	The interval $[\psi(\ell),\ell]$ is the range for $\mathrm{P}_{\text{cr}}(\ell)$.
	Moreover,
	the sequence~$s \in \Per(\ell)$ with period~$(0, 0, \dots, 0, 1)$ satisfies~$\rank(s)=\ell$.
\end{restatable}

\begin{theorem}\label{thm: allowed dimensions}
	Let $\ell \ge 2$.
	Let $S$ be the set of non-empty sets $\{d_1, d_2, \dots, d_j\}$ of non-negative pairwise distinct integers such that $\lcm(d_1, d_2, \dots, d_j) = \ell$. 
	The set $R' = \{\sum_{i = 1}^j \phi(d_i) : \{d_1, d_2, \dots, d_j\} \in S\}$
	is the set of muggle numbers in the range for $\mathrm{P}_{\text{cr}}(\ell)$.
	In particular,
	the set of magic numbers for $\mathrm{P}_{\text{cr}}(\ell)$ is $\N\setminus R'$.
\end{theorem}

The range and the characterization of the magic and muggle numbers with respect to~$\mathrm{P}_{\text{cr}}(\ell)$ from \cref{pro: upper bound tight const-rec} and \cref{thm: allowed dimensions} will follow from intermediate results involving sizes of (simple) matrices with order~$\ell$.
As usual,
we let $\GL_d(\Z)$ denote the \emph{general linear group of degree $d$ over~$\Z$},
made of invertible matrices of size $d$.
Note that for every matrix in $\GL_d(\Z)$,  its determinant is $\pm 1$.
Given a matrix $M\in \GL_d(\Z)$,  its \emph{order} $\ord M$ is the smallest positive integer $m$ such that $M^m=I$,
where~$I$ is the identity matrix of size $d$,  or $+\infty$ if such an integer does not exist.
We say that a matrix $M$ is \emph{simple} if all its eigenvalues (in $\C$) have algebraic multiplicity $1$.

\subsection{Proof of \cref{pro: upper bound tight const-rec}}

In fact,  with \cref{thm: hiller},
Hiller showed that the minimal size of matrices of a given order corresponds to a specific value of the function $\psi$.
Therefore,  it gives the lower bound in \cref{pro: upper bound tight const-rec} (see \seq{A080737,  A152455}).

\begin{theorem}[Hiller~\cite{Hiller}]
	\label{thm: hiller}
	The smallest integer $d \ge 1$ such that $\GL_d(\Z)$ contains a matrix of order $\ell$ is $\psi(\ell)$. 
\end{theorem}

\begin{proof}[of  \cref{pro: upper bound tight const-rec}]
	\label{proof: pro: upper bound tight const-rec}
Every sequence $s$ in $\Per(\ell)$ is constant-recursive since $s(n+\ell)=s(n)$ for all $n \ge 0$.
Hence its rank is at most $\ell$. 
The lower bound on the rank follows directly from \cref{thm: hiller}.
Rank $\ell$ is achieved by the sequence with period $(0, 0, \dots, 0, 1)$.
\qed
\end{proof}

\subsection{Proof of \cref{thm: allowed dimensions}}

The proof of \cref{thm: allowed dimensions} will involve several results on matrices of finite order,  the first of which giving an alternative characterization of muggle numbers with respect to $\mathrm{P}_{\text{cr}}(\ell)$.

\begin{theorem}\label{thm: periodic seq matrices order}
	Let $\ell \ge 2$.
	The set
	\[
		R=\left\{d \in\N : \text{there exists a simple matrix $M \in \GL_d(\Z)$ such that $\ord M = \ell$} \right\}
	\]
	is the set of muggle numbers in the range for $\mathrm{P}_{\text{cr}}(\ell)$.
\end{theorem}

We now develop the necessary tools to prove both \cref{thm: allowed dimensions,thm: periodic seq matrices order}.
%
Let~$s$ be a constant-recursive sequence satisfying the recurrence relation in Equation~\eqref{eq: constant-recursive recurrence}.
The \emph{characteristic polynomial} of the underlying recurrence is $\chi(x) = x^d - c_{d - 1} x^{d-1} - \cdots - c_1 x - c_0$.
For a monic polynomial~$p(x) = x^\ell + a_{\ell - 1} x^{\ell-1} + \cdots + a_1 x + a_0$,  we let~$C(p)$ denote the \emph{companion matrix} of $p(x)$ given by the matrix of size $\ell$
\[
	\begin{bmatrix}
		0 & 0 & \cdots & 0 & -a_0 \\
		1 & 0 & \cdots & 0 & -a_1 \\ 
		0 & 1 & \cdots & 0 & -a_2 \\ 
		\vdots & \vdots & \ddots & \vdots & \vdots \\ 
		0 & 0 & \cdots & 1 & -a_{\ell - 1} \\ 
	\end{bmatrix}.
\]

The following result links the period length of a periodic sequence satisfying Equation~\eqref{eq: constant-recursive recurrence} and the order of the companion matrix of the previous recurrence.

\begin{restatable}{proposition}{PeriodAndOrdMatrix}
	\label{pro: periods for a given recurrence}
	Let $s$ be a constant-recursive sequence satisfying Equation~\eqref{eq: constant-recursive recurrence},
	let $\chi(x)$ be the characteristic polynomial of the underlying recurrence,
	and let $M \in \GL_d(\Z)$ be the companion matrix of $\chi(x)$.
	Assume that $\ord M$ is finite.
	\begin{enumerate}[nosep]
		\item\label{item: period length is at most ord M}
			The period length of $s$ is at most $\ord M$.
		\item\label{item: period length with 0,0,...,0,1 is ord M}
			Moreover, if the initial conditions are $s(0)=s(1)=\cdots=s(d-2)=0$ and $s(d-1)=1$, then the period length of $s$ is  exactly $\ord M$.
	\end{enumerate}
\end{restatable}

The proof of \cref{pro: periods for a given recurrence} was inspired by a similar result for finite fields~\cite[Proposition~IV.5.1]{Rigo}.
We start with a lemma stating that powers of the companion matrix give access to further terms of the sequence.
\begin{lemma}\label{lem: powers of companion matrix}
Let $s$ be a constant-recursive sequence satisfying Equation~\eqref{eq: constant-recursive recurrence} and let $M$ be the companion matrix of the corresponding characteristic polynomial.
For all $n\ge 0$, we define the vector $\s_n =  \begin{bmatrix}
		s(n) & s(n+1) & \cdots & s(n+d-1) 
\end{bmatrix}$. 
Then for all $m,n\ge 0$, we have $\s_n M^m=\s_{n+m}$.
\end{lemma}
\begin{proof}
	We proceed by induction on $m\ge 0$.
	If $m=0$, the result is clear.
	By the induction hypothesis, we have $\s_n M^{m+1}=\s_{n+m} M$, which is equal to
	\[
		\begin{bmatrix}
				s(n+m+1) & s(n+m+2) & \cdots & s(n+m+d-1) & \sum_{i=0}^{d-1} c_i s(n+m+i)
		\end{bmatrix}.
	\]
	By assumption, $\sum_{i=0}^{d-1} c_i s(n+m+i) = s(n+m+d)$, which implies that $\s_n M^{m+1}=\s_{n+m+1}$, as expected.
\qed
\end{proof}

\begin{proof}[of \cref{pro: periods for a given recurrence}]
Let $\ord M = \ell$.
For all $n\ge 0$, \cref{lem: powers of companion matrix} gives $\s_{n+\ell} = \s_n M^\ell = \s_n$, so $s(n+\ell)=s(n)$ for all $n\ge 0$. 
This in turn implies that the period length of $s$ is at most $\ell$ for all initial conditions.
This proves Item~\eqref{item: period length is at most ord M}.

Let us prove Item~\eqref{item: period length with 0,0,...,0,1 is ord M}.
We show that every period length of the sequence $s$ with initial conditions $s(0)=s(1)=\cdots=s(d-2)=0$ and $s(d-1)=1$ is at least $\ell$.
Let $\{\e_i : 1\le i \le d\}$ denote the standard basis of $\R^d$, that is, for all $i\in\{1,2,\ldots, d\}$, $\e_i$ has a $1$ in the $i$th coordinate and $0$'s elsewhere.
By assumption, we have $\s_0=\e_d$.
We can show by induction that, for all $j\in\{0,1,\ldots, d-1\}$,
\[
	\s_j = \s_0 M^j = \e_{d-j} + \sum_{i=d-j+1}^{d} \alpha_{j,i} \e_i
\]
for some coefficients $\alpha_{i,j}\in \Q$.
Let
\[
A
=
\begin{bmatrix}
\s_{d-1} \\
\s_{d-2} \\
\vdots \\
\s_0
\end{bmatrix}
=
\begin{bmatrix}
1 & \alpha_{d-1,2} & \alpha_{d-1,3} & \cdots & \alpha_{d-1,d} \\
0 & 1 & \alpha_{d-2,2} & \cdots & \alpha_{d-2,d} \\ 
\vdots & \ddots & \ddots & \ddots & \vdots \\ 
0 & \cdots & 0 & 1 & \alpha_{1,d} \\ 
0 & \cdots & 0 & 0 & 1 \\ 
\end{bmatrix}
\in \GL_d(\Q).
\]
Note that $\det(A)=1$.
By \cref{lem: powers of companion matrix} again, we obtain
\begin{equation}\label{eq: powers of M with A}
A M^m
=
\begin{bmatrix}
\s_{d-1} \\
\s_{d-2} \\
\vdots \\
\s_0
\end{bmatrix}
M^m
=
\begin{bmatrix}
\s_{m+d-1} \\
\s_{m+d-2} \\
\vdots \\
\s_m \\
\end{bmatrix}
\end{equation}
for all $m\ge 0$.
Now suppose that $s$ belongs to $\Per(m)$.
We want to show that $m\ge \ell$.
By definition of $m$ being a period length, Equation~\eqref{eq: powers of M with A} becomes
\[
A M^m
=
\begin{bmatrix}
\s_{d-1} \\
\s_{d-2} \\
\vdots \\
\s_0
\end{bmatrix}
=
A.
\]
Since $A$ is invertible, we conclude that $M^m =I$, so $m$ is a multiple of $\ord(M)=\ell$.
In particular, $m\ge \ell$, as desired.
\qed
\end{proof}

The following result characterizes matrices of finite order.

\begin{theorem}[{Koo~\cite[page~147]{Koo}}]\label{thm: cyclotomic block-diagonal}
	A matrix $M \in \GL_d(\Q)$ has finite order if and only if there exist positive integers $r, m_1, \dots, m_r, n_1,\ldots,n_r$ (where the $m_i$'s are pairwise distinct), an invertible matrix $P \in \GL_d(\Q)$, and a block-diagonal matrix
	\[
		A = \diag( \underbrace{C(\Phi_{m_1}), \ldots, C(\Phi_{m_1})}_{n_1 \text{ times}}, \underbrace{C(\Phi_{m_2}), \ldots, C(\Phi_{m_2})}_{n_2 \text{ times}}, \ldots, \underbrace{C(\Phi_{m_r}), \ldots, C(\Phi_{m_r})}_{n_r \text{ times}})\text,
	\]
	with $\sum_{i=1}^r n_i\deg(\Phi_{m_i}) = d$ such that $M=P^{-1}AP$.
	Moreover, the order of $M$ is $\lcm(m_1,m_2,\ldots,m_r)$.
\end{theorem}
Observe that $r \le  \sum_{i=1}^r n_i\deg(\Phi_{m_i}) = d$ since $n_i, \deg(\Phi_{m_i})\ge 1$ for all $i$.

We state two lemmas useful to prove \cref{thm: periodic seq matrices order},
describing the characteristic polynomial of the minimal recurrence satisfied by a constant-recursive and periodic sequence.

\begin{restatable}{lemma}{FactorCR}\label{lem: factor}
	Let $\ell \ge 2$.
	For every constant-recursive and periodic sequence~$s$ whose period length divides~$\ell$,
	the characteristic polynomial of the minimal recurrence satisfied by~$s$ divides $x^\ell - 1$.
\end{restatable}
\begin{proof}
	\label{proof: lem: factor}
By assumption,  we have $s(n+\ell)=s(n)$ for all $n\ge 0$.
The characteristic polynomial of the latter recurrence is $x^\ell-1$.
The conclusion follows by minimality.
\qed
\end{proof}

\begin{restatable}{lemma}{FiniteOrder}
	\label{lem: finite order}
	Let $\ell \ge 2$ and let $f(x)$ be a factor of $x^\ell - 1$.
	The companion matrix $M$ of $f(x)$ has finite order.
	Moreover, $f(x) = \prod_{i = 1}^r \Phi_{m_i}(x)$ is a product of distinct cyclotomic polynomials where each index $m_i$ is a divisor of $\ell$ and $\ord M = \lcm(m_1, m_2, \dots, m_r)$.
\end{restatable}
\begin{proof}
	\label{proof: lem: finite order}
Since the roots of $f(x)$ are roots of unity, the eigenvalues of $M$ are roots of unity with order dividing $\ell$.
Therefore,  the matrix $M$ has finite order by \cref{thm: cyclotomic block-diagonal}, and its order is $\lcm(m_1, m_2, \dots, m_r)$.
\qed
\end{proof}

We are now ready to prove \cref{thm: periodic seq matrices order}.

\begin{proof}[of \cref{thm: periodic seq matrices order}]
	For the sake of conciseness,
	set~$T = \left\{\rank(s) : s \in \Per(\ell)\right\}$.
	We prove that $T=R$ by showing two inclusions.

	\textbf{First inclusion.}
		We show $T \subseteq R$.
		So let $s \in \Per(\ell)$ and let $d = \rank(s)$.
		We construct a simple matrix $M \in \GL_d(\Z)$ such that $\ord M = \ell$.
		Since $\rank(s)=d$, the sequence $s$ satisfies a recurrence of the form of Equation~\eqref{eq: constant-recursive recurrence} by definition.
		Let $M$ be the companion matrix of the characteristic polynomial $\chi(x)$ of this recurrence.
		By \cref{lem: factor}, $\chi(x)$ is a factor of $x^\ell - 1$, so the matrix $M$ is simple.
		By \cref{lem: finite order}, $M$ has finite order and $\ord M \leq \ell$.
		The other inequality, $\ell \le \ord M$, follows from \cref{pro: periods for a given recurrence}.
		Therefore, $T \subseteq R$.
		
	\textbf{Second inclusion.}
		We show $R \subseteq T$.
		Let $d \ge 1$ and let $M \in \GL_d(\Z)$ be a simple matrix such that $\ord M = \ell$.
		By minimality,  for all $m \in \Divs(\ell)$ with $m < \ell$, we have $M^m\neq I$.
		By \cref{thm: cyclotomic block-diagonal}, $\chi_M(x) = \prod_{i=1}^r \Phi_{m_i}(x)^{n_i}$ for some integers $r, m_i,n_i$, where each $m_i$ is a divisor of $\ell$, the $m_i$'s are distinct, and $\lcm(m_1, m_2, \dots, m_r) = \ell$.
		Since $M$ is simple,  $n_i = 1$ for all $i$.
		Write $\chi_M(x) = x^{d} - c_{d - 1} x^{d - 1} - \dots - c_1 x - c_0$ where $c_0, c_1, \dots, c_{d - 1} \in \Z$.
		Let $s$ be the sequence satisfying the recurrence $s(n + d) = c_{d - 1} s(n + d - 1) + \dots + c_1 s(n + 1) + c_0 s(n)$
		for all $n \ge 0$ with initial conditions $s(0)=s(1)=\cdots=s(d-2)=0$ and $s(d-1)=1$.
		By \cref{pro: periods for a given recurrence},  the period length of $s$ is $\ord(M)=\ell$,  so $s\in\Per(\ell)$.
		We claim that $\rank(s)=d$.
		By definition,  $s$ has rank at most $d$. 
		For all $d' < d$,  the only sequence beginning with $d'$ initial zeroes and satisfying a recurrence  relation of order $d'$ is the constant-zero sequence.
		Hence $\rank(s)=d$.
		Consequently,  $R \subseteq T$.
\qed
\end{proof}

We illustrate each direction of the proof with an example.

\begin{example}
	Let $s\in\Per(3)$ be the sequence with period $(-1, 0, 1)$.
	It satisfies $s(n + 2) = -s(n + 1) - s(n)$ for all $n \ge 0$.
	The characteristic polynomial of this recurrence is $x^2 + x + 1$,
	whose companion matrix is
	\[
		M = \begin{bmatrix}
				0 & -1 \\
				1 & -1
			\end{bmatrix}
	\]
	and which is a factor of $x^3 - 1$.
	It can be checked that $\ord M$ indeed equals $3$.
\end{example}

\begin{example}
	Let let $M\in \GL_6(\Z)$ be the matrix whose rows are given by the six vectors
	\[
		\begin{array}{ccc}
			\begin{bmatrix}
				1 & -1 & 1 & 0 & 1 & 1
			\end{bmatrix}\text,
			&
			\begin{bmatrix}
				1 & -1 & 1 & 0 & 1 & 0 
			\end{bmatrix}\text,
			&
			\begin{bmatrix}
				2 & -1 & 1 & 1 & 0 & 0 
			\end{bmatrix}\text, \\
			\begin{bmatrix}
				-2 & 1 & -2 & -1 & -1 & -1
			\end{bmatrix}\text,
			&
			\begin{bmatrix}
			-1 & 0 & 0 & 0 & -1 & 0
			\end{bmatrix}\text,
			&
			\begin{bmatrix}
				-1 & 1 & -1 & 0 & 0 & -1
			\end{bmatrix}\text.
		\end{array}
	\]
	One can check that $M$ has order $15$ and its characteristic polynomial is equal to $\chi_M(x) 
	= x^6+2 x^5+3 x^4+3 x^3+3 x^2+2 x+1
	= \Phi_3(x)\Phi_5(x)$.
	It is not difficult to find an invertible matrix $P \in \GL_6(\Q)$ such that $P M P^{-1} = \diag(C(\Phi_3),C(\Phi_5))$.
	We define $s$ to be the sequence satisfying $s(0)=s(1)=s(2)=s(3)=s(4)=0$, $s(5)=1$,  and,  for all $n \ge 0$,
	\[
		s(n + 6) = -2s(n + 5) -3s(n + 4) -3s(n + 3) -3s(n + 2) -2s(n + 1) -s(n)\text.
	\]
	It can be verified that $s$ has period $(0, 0, 0, 0, 0, 1, -2, 1, 1, -2, 2, -1, -1, 2, -1)$,  so $s\in\Per(15)$ and $\rank(s)=6$.
\end{example}

We now turn to the proof of \cref{thm: allowed dimensions}.
Note that the $\lcm$ condition in the statement implies that each $d_i$ is a divisor of $\ell$.


\begin{proof}[of \cref{thm: allowed dimensions}]
	Let $\ell \ge 2$.
	By \cref{thm: periodic seq matrices order},
	the sets $ \left\{\rank(s) : s \in \Per(\ell)\right\}$ and $R$ are equal.
	By \cref{thm: cyclotomic block-diagonal} and \cref{lem: degree cyclotomic},
	there exists a simple matrix $M$ of size $d$ with order~$\ell$ if and only if $d = \sum_{i = 1}^j \phi(d_i)$ for some set of integers $\{d_1, d_2, \dots, d_j\} \in S$.
	The conclusion of the first part follows and the second one follows by \cref{pro: upper bound tight const-rec}.
\qed
\end{proof}

Notice that \cref{thm: allowed dimensions} directly implies the following,  explaining why only two values are present on prime rows in \cref{tab:values-rank-cr}.

\begin{corollary}
	For every prime~$\ell$,  the only muggle numbers with respect to $\mathrm{P}_{\text{cr}}(\ell)$ are~$\ell-1$ and~$\ell$.
\end{corollary}

\section{Regular Sequences}\label{sec: regular sequences}

In this section, we look at $k$-regular sequences~\cite{Allouche--Shallit-1992}, which are base-$k$ analogues of constant-recursive sequences.
An integer sequence $s$ is \emph{$k$-regular} if the $\Q$-vector space $V_k(s)$ generated by its $k$-kernel (defined in \cref{section: automatic sequences}) is finitely generated.
The \emph{$k$-rank} (or simply the \emph{rank} when the context is clear) of $s$ is the dimension of $V_k(s)$.
We use, again, the notation $\rank_k(s)$.

Every periodic sequence is $k$-automatic for every $k$ and thus also $k$-regular.
However, the corresponding $k$-ranks of are not necessarily equal.

\begin{example}\label{ex: aut VS reg}
	Consider the sequence $s\in\Per(4)$ with period $(0,1,1,1)$.
	The $2$-kernel of~$s$ consists of the sequences with periods~$(0,1,1,1)$,  $(0,1)$, $(1)$, and~$(0)$.
	Viewed as a $2$-automatic sequence,
	$s$ has rank $4$,
	while, when seen as a $2$-regular sequence,
	it has rank $3$ (the sequences with periods~$(0,1,1,1)$,  $(0,1)$,  and $(1)$ are linearly independent).
\end{example}

For integers $k,\ell\ge 2$ and the family of periodic sequences $\Per(\ell)$,
we consider the following property:
an integer satisfies $\mathrm{P}_{\text{r}}(k,\ell)$ if and only if it is equal to the $k$-rank of some sequence in $\Per(\ell)$.
%
%
%
%
To look for magic and muggle numbers with respect to $\mathrm{P}_{\text{r}}(k,\ell)$,  we will use the following notion.
%
Let $n\ge 1$ be an integer and let $(a_0, a_1,\ldots, a_{n-1})$ be a sequence of integers.
The \emph{circulant matrix} $C(a_0, a_1,\ldots, a_{n-1})$ is the matrix of size $n$ for which the $i$th row,  $i\in\{1,\ldots,n\}$,  is the $(i-1)$st shift of the first row.
The polynomial $f(x) = \sum_{i=0}^{n-1} a_i x^i$ is the \emph{associated polynomial} of~$C(a_0, a_1,\ldots, a_{n-1})$.

We link the rank of the circulant matrix associated with a given periodic sequence to its constant-recursive and $k$-regular ranks.
The proof of the first result follows from~\cite[Theorem~2.1.6]{Berstel-Reutenauer} and the discussion in~\cite[Chapter~6,  Section~1]{Berstel-Reutenauer}.

\begin{proposition}[{Berstel \& Reutenauer~\cite{Berstel-Reutenauer}}]
	\label{pro: rankCR equal rank circulant}
	Let $\ell\ge 2$. 
	Let $s \in \Per(\ell)$ and let $M=C(s(0),s(1),\ldots,s(\ell-1))$ be the circulant matrix of size $\ell$ built on the period of $s$.
	The rank of $s$ as a constant-recursive sequence is equal to $\rank M$.
\end{proposition}

\begin{restatable}{proposition}{EqualityRanksRegularCirculant}
	\label{pro: equality ranks reg and circulant}
	Let $k,\ell\ge 2$ be coprime integers.
	Let $s \in \Per(\ell)$ and let $M=C(s(0),s(1),\ldots,s(\ell-1))$ be the circulant matrix of size $\ell$ built on the period of~$s$.
	Then $\rank_k(s)=\rank M$. 
\end{restatable}
\begin{proof}
By iteratively applying \cref{lem: subsequence coprime},  every sequence $s(k^e n + j)$ with $e\ge 0$ and $0\le j \le k^e - 1$ has period length $\ell$,  so we have $\rank_k(s),\rank M \le \ell$.
As a first easy case,  if the sequences $s(n)\n, s(n+1)\n,\ldots,s(n+\ell-1)\n$ are linearly independent,  then $\rank M = \ell = \rank_k(s)$ by definition,  which suffices.
As a second case,  assume that there is a smallest set of linearly independent sequences among them,  say of size $d < \ell$.
By definition,  we have $\rank M = d$.
Using \cref{pro: rankCR equal rank circulant},  every sequence in
\[
V = \{ s'\in\Per(\ell) \colon s' \text{ satisfies the same linear recurrence as } s \}
\]
can be written as a linear combination of the $d$ linearly independent sequences.
We now prove that every sequence in $\ker_k(s)$ belongs to $V$.
By \cref{lem: rewriting},  such a sequence can be written as $s(cn+r)\n$ with $0\le c,r\le \ell-1$ (and $c$ is invertible in $\Z_\ell$).
Since the sequences~$s$ and~$s(cn+r)\n$ both belong to~$\Per(\ell)$,
let~$f$ and~$g$ be the respective characteristic polynomials of their minimal recurrences.
By \cref{lem: factor} and the factorization of $x^{\ell}-1$ into cyclotomic polynomials,  there exists a set $S\subseteq \Divs(\ell)$ such that $f(x) =\prod_{m \in S} \Phi_m(x)$.
The roots of $g(x)$ are the~$c$th powers of those of $f(x)$.
Since $k$ and $\ell$ are coprime,  so are $c$ and $m$,  thus the map $x \mapsto x^c$ is a permutation of the roots of $\Phi_m(x)$.
Therefore,  the roots of~$g$ are equal to those of $f$.
So $s(cn+r)\n$ belongs to $V$.
Therefore,  $\rank_k(s) = d$,  as desired.
\qed
\end{proof}

Note that the previous argument does not apply to the case where $k$ and $\ell$ are not coprime.

\begin{example}
In \cref{ex: aut VS reg},  the kernel sequence $s$ with period $(1,0)$ is not a linear combination of the other kernel sequences having respective periods $(0,1,1,1)$, $(1)$,  and $(0)$.
However,
$s$ is a linear combination of the shifts of sequence $s_0$ with period $(0,1,1,1)$ since,
if we let $s_1$,  $s_2$, and $s_3$ respectively be the first,  second, and third shifts of $s_0$,
we have $s=\frac{2}{3} s_0 - \frac{1}{3} s_1 + \frac{2}{3} s_2 - \frac{1}{3}s_3$. 
\end{example}

%
%
%

\begin{theorem}\label{thm: characterization regular sequences coprime case}
	Let $k,\ell\ge 2$ be coprime integers.
The range for $\mathrm{P}_{\text{r}}(k,\ell)$ is~$[\psi(\ell),\ell]$.
Now let $S$ be the set of non-empty sets~$\{d_1, d_2, \dots, d_j\}$ of non-negative pairwise distinct integers with $\lcm(d_1, d_2, \dots, d_j) = \ell$. 
	The set $R' =\{\sum_{i = 1}^j \phi(d_i) : \{d_1, d_2, \dots, d_j\} \in S\}$
	is the set of muggle numbers in the range for $\mathrm{P}_{\text{r}}(k,\ell)$.
	In particular,  the set of magic numbers with respect to $\mathrm{P}_{\text{r}}(k,\ell)$ is $\N\setminus R'$.
\end{theorem}
\begin{proof}
	In this proof,
	for every sequence $s \in \Per(\ell)$,
	we let $M_s$ denote the circulant matrix $C(s(0),s(1),\ldots,s(\ell-1))$ of size $\ell$ built on the period of $s$.
	As in \cref{sec: constant-recursive sequences},  we also let $\rank(s)$ denote the rank of $s$ viewed as constant-recursive sequence.
	By \cref{pro: rankCR equal rank circulant,pro: equality ranks reg and circulant},  we have that 
	\[
		\left\{\rank_k(s) : s \in \Per(\ell)\right\}
		=
		\left\{\rank M_s : s \in \Per(\ell)\right\}
		=
		\left\{\rank(s): s \in \Per(\ell)\right\}\text.
	\]
	In particular,  the sets of muggle numbers for $\mathrm{P}_{\text{r}}(k,\ell)$ and $\mathrm{P}_{\text{cr}}(\ell)$ are equal.
	\cref{pro: upper bound tight const-rec,thm: allowed dimensions} then allow to conclude the proof.
\qed
\end{proof}

\section{Conclusion and Future Work}

In this paper,  we studied the magic number problem for three definitions of the rank of periodic sequences.
In particular,
in the case of constant-recursive sequences,
we obtain a full characterization of magic numbers for every period length,
while,
for~$k$-automatic and $k$-regular sequences with period length~$\ell$,
we give a characterization provided that~$k$ and~$\ell$ are coprime.
For such sequences,
the general case looks more intricate.
Indeed,
as \cref{ex: aut VS reg} already shows,
the period lengths of sequences in the $k$-kernel are not always equal to $\ell$,
but rather divide $\ell$.
Therefore,
a deep understanding of the relationship between the sets $\Per(d)$, where $d$ divides $\ell$,
is required to characterize the possible ranks,
and so in turn the magic and muggle numbers within these families of sequences.
Hence,
as a natural pursuance of this study,
we are currently working on the generalization of the results for automatic and regular sequences to  
the case where no condition on~$k$ and~$\ell$ is specified.

\subsubsection{Acknowledgments}

Savinien Kreczman and Manon Stipulanti are supported by the FNRS Research grants 1.A.789.23F and 1.B.397.20F respectively.

\newpage

\bibliographystyle{splncs04}
\bibliography{biblio}

\end{document}